\documentclass[submission,copyright,creativecommons]{eptcs}
\providecommand{\event}{AREA 2022} % Name of the event you are submitting to
\usepackage{breakurl}             % Not needed if you use pdflatex only.
\usepackage{underscore}           % Only needed if you use pdflatex.

\usepackage[utf8]{inputenc}
\usepackage[T1]{fontenc}
\usepackage[english]{babel} % or french or ...
\usepackage{amsthm}
\usepackage{amssymb,amsmath}
\usepackage{graphics}
\usepackage{graphicx}
\usepackage{subcaption}
\usepackage{algorithm} 
\usepackage{algorithmic}
\usepackage{tikz}
\usetikzlibrary{decorations,arrows,shapes,automata}
\usepackage{enumerate}
\usepackage{thmtools}

\mathchardef\mhyphen="2D

\usepackage[colorinlistoftodos]{todonotes}

\renewcommand\thmcontinues[1]{continued}

\title{Careful Autonomous Agents in Environments With Multiple Common Resources}

\author{
Rodica Condurache 
\institute{``Alexandru Ioan Cuza'' University of Ia\c{s}i, Romania}
\email{rodica.condurache@info.uaic.ro } 
\and
Catalin Dima
\institute{Univ.~Paris Est Cr\'eteil, LACL, F-94010 Creteil, France}
\email{dima@u-pec.fr}
\and
Madalina Jitaru
\institute{Univ.~Paris Est Cr\'eteil, LACL, F-94010 Creteil, France}
\email{madalina.jitaru@u-pec.fr}
\and
Youssouf Oualhadj
\institute{Free University of Bozen-Bolzano, Italy\\
Univ. Paris Est Cr\'eteil, LACL, F-94010 Creteil, France}
\email{youssouf.oualhadj@unibz.it}
\and
Nicolas Troquard
\institute{Free University of Bozen-Bolzano, Italy}
\email{nicolas.troquard@unibz.it}
}
\def\titlerunning{Careful Autonomous Agents in Environments With Multiple Common Resources}
\def\authorrunning{R. Condurache, C. Dima, M. Jitaru, Y. Oualhadj, N. Troquard}

\newcommand{\states}[0]{\mathsf{S}}
\newcommand{\uStates}[0]{\widehat{\states}}
\newcommand{\init}[0]{s_0}
\newcommand{\edges}[0]{\mathsf{E}}
\newcommand{\uEdges}[0]{\widehat{\edges}}
\newcommand{\game}{\mathcal{G}}
\newcommand{\uGame}{\widehat{\game}}

\newcommand{\src}{\mathsf{Src}}
\newcommand{\trgt}{\mathsf{Trgt}}
\newcommand{\pref}[2]{#1[..#2]}
\newcommand{\suff}[2]{#1[#2..]}
\newcommand{\plays}[0]{\mathsf{Plys}}
\newcommand{\hist}[0]{\mathsf{Hst}}
\newcommand{\last}[0]{\mathsf{lst}}
\newcommand{\set}[1]{\{ #1 \}}
\newcommand{\play}[1]{\out{#1}}
\newcommand{\out}[1]{\langle #1 \rangle}
\newcommand{\obj}{\mathsf{Obj}}

\newcommand{\cst}{\mathsf{cost}}
\newcommand{\bcst}{\overline{\mathsf{cost}}}

\newcommand{\bcstfull}{\bcst \colon \edges \to \bbZ^d}
\newcommand{\en}{\mathsf{Energy}}
\newcommand{\multiEn}{\mathsf{MultiEnergy}}

\newcommand{\gameful}{
  \game = 
  \langle\states, (\states_1 \uplus \ldots \uplus\states_{n}), 
  \init, \players, \edges, \ap, \lab\rangle  
}
\newcommand{\players}[0]{\mathsf{P}}
\newcommand{\lab}[0]{\ell}
\newcommand{\prof}[0]{\overline{\sigma}}
\newcommand{\uProf}[0]{\widehat{\prof}}

\newcommand{\fullProf}[0]{\prof = \langle \sigma_1,\ldots, \sigma_n \rangle}
\newcommand{\devProf}[1]{\langle #1 \rangle}
\newcommand{\pay}{\mathsf{Payoff}}
\newcommand{\coa}[1]{{\mhyphen#1}}
\newcommand{\NE}[0]{\mathsf{NE}}

\newcommand{\winone}{W_1}
\newcommand{\wintwo}{W_2}

\newcommand{\bbN}{\mathbb{N}}
\newcommand{\bbZ}{\mathbb{Z}}

\newcommand{\twoexptime}[0]{\mathsf{2EXPTIME}}

\newcommand{\untl}{~\mathsf{U}~}

\newcommand{\nxt}{\mathsf{X}~}
\newcommand{\ap}{\mathsf{AP}}
\newcommand{\true}{\mathsf{True}}

\newtheorem{theorem}{Theorem}
\newtheorem{lemma}[theorem]{Lemma}
\newtheorem{proposition}[theorem]{Proposition}

\newtheorem{example}[theorem]{Example}

\newif{\ifMarginalComments}
\MarginalCommentstrue
\newcommand{\camready}[1]{#1}

\begin{document}
\maketitle

\begin{abstract}
\camready{Careful rational synthesis was defined in  \cite{DBLP:conf/atal/ConduracheDOT21}  as a quantitative extension of Fisman et al.'s rational synthesis \cite{FKL10}, as a model of multi-agent systems} in which agents are interacting in a graph arena in a turn-based 
fashion. There is one common resource, and each action may decrease or increase the resource. Each agent has a temporal qualitative objective and wants to maintain the value of the 
resource positive. One must find a Nash equilibrium. This problem is decidable.
%\todo{The N.E. has to satisfy a global property}

In more practical settings, the verification of the critical properties of multi-agent 
systems calls for models with many resources. Indeed, agents and robots consume and 
produce more than one type of resource: electric energy, fuel, raw material, 
manufactured goods, etc.
We thus explore the problem of careful rational synthesis with several resources.
We show that the problem is undecidable.
We then propose a variant with bounded resources, motivated by the observation that in practical settings, the storage of resources is limited. We show that the problem becomes decidable, and is no harder than controller synthesis with Linear-time Temporal Logic \camready{objectives}.
\end{abstract}

\section{Introduction}
%\todo{This paper is great for \url{https://areaworkshop.github.io/AREA2022/}! Here's why.}
The presence of autonomous agents in modern societies has become commonplace. We interact with them every day, and they may be of different levels of autonomy, e.g., self-checkout, chatbots, robot vacuum cleaners, or virtual assistants.
A current tendency is that agents are intruding on the physical world, and robots are expanding their territory beyond their confined industrial environment.

\medskip
The access to the resources necessary for an agent to accomplish his tasks could have been simply assumed in many application domains before: direct wire to an electricity source, a human operator providing raw material, etc.
Nowadays, typical agents must be more autonomous than before in managing the multiple resources they need. They must carefully consume them, and in presence of competitors, they must also be careful in how they produce them.

\medskip

Linear-time Temporal Logic (LTL)~\cite{DBLP:conf/focs/Pnueli77} has been a very popular logic for specifying temporal properties of systems. 
Planning with objectives expressed in some temporal logic has been well studied~\cite{BACCHUS2000123,Doherty_Kvarnstram_2001,DBLP:conf/kr/BienvenuFM06,CIMATTI2008841}.
Some logics have also been proposed to explicitly verify the properties of multiagent systems in presence of resource constraints~\cite{DBLP:journals/corr/abs-1303-0789,DBLP:journals/logcom/NguyenALR18,DBLP:conf/birthday/AlechinaL20}.
When agents roam more freely the physical world, they are more likely to compete with other agents, human or artificial, which may have conflicting goals. When planning in such environments an agent needs to adapt his behaviour to the capabilities and goals of others. A solution to a multiagent planning problem in this setting is a non-cooperative strategic equilibrium: a vector of strategies, one for each agent, such that no individual agent can be better off by unilaterally changing their strategy. This is what has come to be known as a Nash equilibrium~\cite{Nash51}.

\medskip
This paper aims to contribute to the line of research interested in the formal verification of the existence of Nash equilibria in a multiagent 
system~\cite{ummels08,FKL10,CFGR16,DBLP:journals/apin/AbateGHHKNPSW21}. When there is a solution Nash equilibrium, the techniques used can actually return a multiagent plan that satisfies the requisites. %\todo{we actually build one, which is far better than just checking the existence}
The paper has a special focus to consider agents that must be autonomous in an environment with multiple common resources, to bring the theory closer to the reality that engineers are working with. 

\medskip

In~\cite{DBLP:conf/atal/ConduracheDOT21}, the problem of careful rational synthesis is defined as a  quantitative variant of rational synthesis~\cite{FKL10}.
Agents interact in a graph arena in a turn-based fashion. Each state is controlled by one and only one agent who decides which edge to
follow. Each agent has a temporal objective that he tries to achieve. 
There is one integer common resource, and each action may decrease or increase the resource. 
The rational synthesis problem consists in computing a Nash equilibrium that satisfies a 
system objective.
It is shown that in presence of one common resource, deciding the existence of a strategic equilibrium for careful autonomous agents (with \emph{parity} objectives, a canonical representation of temporal properties on infinite traces~\cite{Farwer2002}) can be solved in polynomial space. With LTL objectives, the problem can be solved in doubly exponential space.

\medskip
But in real-case scenarios, physical agents are operating in a world where there is more than one resource. 
In this paper, we explore the problem of careful rational synthesis with several common resources.
\begin{example}[label=ex] Consider the game with $2$ resources illustrated on Figure~\ref{fig:example}.
Players~$1$, $2$ and $3$ control the states $a$, $b$, and $c$ respectively.
The other states are controlled by Player~1 \camready{(but note that the agent who controls them is irrelevant)}. Player~$1$ wants to reach a state with $\bigcirc$, Player~$2$ wants the reach a state with $\Box$, Player~$3$ wants to reach a state with $\Diamond$. All of them want to keep the resources in check: they would be dissatisfied if any of the resources were to go below zero.
The objective of the system is $\bigcirc$. A solution to the synthesis problem is thus a Nash equilibrium that reaches the state $(\bigcirc, \Box)$, and never depletes the resources.

One starts with the resources being $(0,0)$.
Player~$1$ must pump thrice on $a$, which brings the resources to $(6,3)$. (Only he can increase resource two, and at least an amount of $3$ is necessary to reach his objective and the objective of the system.) Player~$1$ can then go to $b$, which brings the resources to $(6,2)$.

At that point, Player~$2$ could go down. 
%There is a non-solution Nash equilibrium here. 
This would be the outcome of a Nash equilibrium, but it would not be a solution to our synthesis problem since we are seeking an equilibrium 
satisfying the system's objective. Instead,
let Player $2$ go to $c$; this brings the resources vector to $(4, 1)$.

At that point, Player~$3$ can go down. 
%There is a non-solution Nash equilibrium here, where both Player~$2$ and $3$ are happy. 
Once again this is the outcome of a Nash equilibrium, but this would not be solution.
Instead, Player~$3$ could go right, 
and the run so obtained would satisfy the objective of the system and keep the resources in check. However this is not the outcome of a Nash equilibrium since Player~$3$ can 
deviate and increase his payoff by going down.
%but this would not be a profitable deviation for him.

In fact, there is no solution to the synthesis problem.
\end{example}

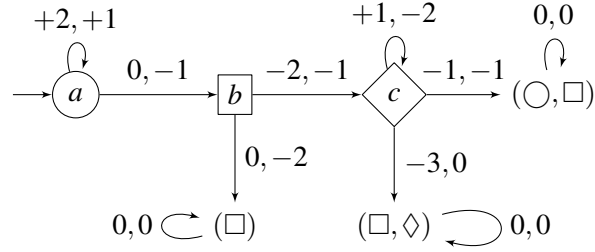
\begin{figure}
\centering
\begin{tikzpicture}[>=latex', join=bevel, initial text = , every node/.style=]
  % States                                                                                                                                                                               
  \node (0) at (0bp, 0bp) [initial left, draw, circle]{$a$};
  \node (1) at (60bp, 0bp) [draw,  rectangle]{$b$};
  \node (11) at (60bp, -50bp) []{$(\Box)$};
  \node (2) at (120bp, 0bp) [draw, diamond]{$c$};
  \node (21) at (120bp, -50bp) []{$(\Box, \Diamond)$};
  \node (22) at (180bp, 0bp) []{$(\bigcirc, \Box)$};
  % Edges                                                                                                                                                                                
  \draw[->] (0) to node [above] {$0, -1$} (1);
  \draw[->] (1) to node [above] {$-2, -1$} (2);
  \draw[->] (1) to node [right] {$0, -2$} (11);
  \draw[->] (2) to node [right] {$-3, 0$} (21);
  \draw[->] (2) to node [above] {$-1, -1$} (22);
  % Selfloops                                                                                                                                                                            
  \draw (0) edge [loop above] node [above] {$+2, +1$} (0);
  %\draw (1) edge [loop above] node [above] {$+1, 0$} (1);
  \draw (11) edge [loop left] node [left] {$0,0$} ();
  \draw (22) edge [loop above] node [above] {$0,0$} ();
  \draw (21) edge [loop right] node [right] {$0,0$} ();
  \draw (2) edge [loop above] node [above] {$+1, -2$} (1);

\end{tikzpicture}
\caption{\label{fig:example} A $2$-resource $3$-player game.}
\end{figure}

It is unfortunately a negative result that we must report in Section~\ref{sec:undecidability}. Deciding the existence of a strategic equilibrium for careful autonomous agents in environments with multiple common resources is indeed undecidable.

We then propose in Section~\ref{sec:bounded-resources} a variant with bounded resources. In this setting, every resource has a maximum capacity.
%

%\begin{example}[continues=ex]
\begin{example}
Suppose now that both resources are bounded with bounds $(3,3)$.
% NOTE TO US: ALSO A SOLUTION WITH BOUNDS (2, 5)

As before, Player~$1$ must pump thrice on $a$, with the resources values being successively $(2,1)$, $(3,2)$, and $(3,3)$. 
As before, Player~$2$ could win in $b$ by going down, and again 
%a non-solution Nash equilibrium is here. 
this would be the outcome of a Nash equilibrium but would not be a solution.
Instead, let Player~$2$ go to $c$, which brings the resources to $(1, 2)$.

\camready{To achieve his goal, in $c$,} Player~$3$ must go down.
\camready{To keep the first resource above zero, he must pump on $c$ twice, thus bringing the value of the first resource to $3$. But doing so
%, the amount of resource two goes below zero. So,
he would deplete the second resource. If Player~$3$ instead carefully moves to the right, Player~$1$ and Player~$2$ meet their objectives, and so does the system.
Hence this outcome results in a Nash equilibrium, that is, a solution.}
%although there is a non-solution Nash equilibrium here, this is not a profitable strategy, wrt.\ the one of going right to $(\bigcirc, \Box)$. He is unhappy going down because he 
%does not satisfy his objective of not depleting the resources; he is unhappy going right because he does not satisfy his qualitative objective.
%If the game goes to $(\bigcirc, \Box)$, Player~$1$ and Player~$2$ are both happy, and thus have no incentive to change their strategy.

To summarize, when the resources are bounded with bounds $(3,3)$, the strategies of Player~$1$ taking the self-loop thrice, then going to $b$, Player~$2$ going to $c$, and Player~$3$ going to $(\bigcirc, \Box)$, form a Nash equilibrium which is a solution to the synthesis problem.
\end{example}

This variant with bounded resource storage capacity is of interest for the practical engineering of autonomous multiagent systems for two reasons. 
The first reason is conceptual. In many real-case scenarios, resources are bounded: e.g., in a community, a shared tank of water can only contain a predetermined amount of water, a shared microgrid powerpack can only contain a predetermined amount of energy, etc.
The second reason is algorithmic. We will show that unlike in the setting with unbounded resources, the problem of rational synthesis in this bounded setting becomes decidable.
Even better, with objectives expressed in LTL, it is not harder than the plain reactive synthesis problem, which is $\twoexptime$-complete~\cite{PR89}.

%%%%%%%%%%%%%%%%%%%%%%%%%%%%%%%%%%%%%%%%%%%%%%%%%%%%%%%%%%%%%%%%%%%%%%%%

\section{Games on finite graphs}
\label{sec:games}

For any set $Q$ we denote by $Q^*$ the set of finite sequences of elements
in $Q$ and $Q^\omega$ the set of infinite sequences of elements of $Q$.
Let $w\in Q^*\cup Q^\omega$, and $i\ge 1$, 
we denote by $w[i]$ the $i$-th element in $w$; we
denote by $\pref{w}{i}$ the prefix of $w$ of size $i$ and $\suff{w}{i}$
the suffix that starts at the $i$-th letter. For an element $q\in Q^*$,
$\last(q)$ is the last element in the sequence $q$.

\subsection{Arenas, strategies and profiles}
\paragraph{Multi-player arenas}
A \emph{multi-player arena} is a tuple
$\gameful$, where $\states$ is a finite
set of states, $(\states_1 \uplus \ldots\uplus \states_{n})$ is a
partition of $\states$, $\init$ is an initial state,
$\players = \set{1,\ldots, n}$ is the set of players,
$\edges$ is an edge relation
in $\states\times\states$, $\ap$ is the set of labels (atomic propositions), and $\lab \colon \states \to 2^{\ap}$ is the labeling function. 
%%(each state is labeled with a set of propositions)%%
For every edge $e = (s,t), \src(e)$ is $s$ and $\trgt(e)$ is $t$.

\paragraph{Plays and strategies}
For an arena $\game$, we denote the plays of this game by $\plays(\game)$ that is
the set of elements $\init s_1 s_2 \ldots$ in $\states^\omega$
such that for all $i \ge 0$, $(s_i,s_{i+1})$ is in $\edges$.
The set $\hist(\game)$ is the set of prefixes of elements in $\plays(\game)$.
Moreover $\hist_i(\game)$ for $i$ in $\players$ is the set of elements
in $\hist(\game)$ whose last element is in $\states_i$:
\[
  \hist_i(\game) =
  \set{
    h \in \hist(\game) \mid
    \last(h) \in \states_i
  }
\]
A \emph{strategy} for player~$i$ is a function
\[
  \sigma_i \colon \hist_i(\game) \to \states 
\]
mapping a history whose last element is $s$ to a state $s'$ such that $(s, s') \in \edges$.
For a strategy $\sigma_i$ for player~$i$, 
we define the set $\out{\sigma_i}$ as the
set of plays that are \emph{compatible} with $\sigma_i$ i.e.,
\begin{align*}
  \out{\sigma_i} =
  \set{
    \pi \in \plays(\game) 
    \mid
    \forall j \ge 0,\ 
    \pref{\pi}{j} \in \hist_i(\game)
    \implies
    \sigma_i(\pref{\pi}{j}) = \pi[j+1]
  }
\end{align*}

\paragraph{Profile of strategies}
Once a strategy $\sigma_i$ for each player $i$ is chosen,
we obtain a strategy profile $\fullProf$. 
\camready{Note that a strategy profile has a unique play in its outcome.}
% By $\prof[i]$ we mean the strategy $\sigma_i$ of player $i$.
$\prof_{\coa{i}}$ is the corresponding partial profile
without the strategy for player $i$.
For a strategy $\sigma'_i$ for a player $i$, we write
$\devProf{\prof_{\coa{i}}, \sigma'_i}$ the profile
$\devProf{\sigma_1, \ldots, \sigma'_i, \ldots, \sigma_n}$.
%We keep the same definitions for $\out{\sigma_i}$ and 
We denote by $\play{\prof}$ the unique outcome of the strategy profile $\prof$.

\subsection{Objectives and payoffs}
An objective $\obj$ is a subset of $\plays(\game)$.
We write $\obj_i$ to specify that it is the objective of player $i$.
We define the payoff $\pay_i(\prof)$ of player $i$ wrt.\ the profile $\prof$ as follows:
\begin{align*}
  \pay_i(\prof) =
  \begin{cases}
    1 \text{ if } \play{\prof} \in \obj_i\\
    0 \text{ otherwise}
  \end{cases}
\end{align*}

% In the case where the arena consists of only two players, we can
% define zero-sum objectives, i.e.\ objectives that oppose for the
% players, formally
% \begin{align*}
%   \forall i \in \set{1,2},\
%   \obj_{3-i} = 
%   \plays(\game) \setminus \obj_i
% \end{align*}

% Once an arena $\game$ is equipped with an objective $\obj_i$ for each
% player $i$,  we will often call \emph{game} the tuple
% $\langle\game,\obj_1,\ldots,\obj_n\rangle$. When the
% objective is clear from the context we will simply write $\game$.

% Thanks to the zero-sum nature, we can define the notion of
% a \emph{winning} strategy for player $i$, i.e., a strategy $\sigma_i$
% s.t. $\out{\sigma_i}$ is a subset of $\obj_i$.

% We write $\game^{i,\coa{i}}$ for the zero-sum game where 
% player~$1$ is $i$, and player 2 is the coalition of the rest of players seen
% as one entity.
% %\todo{Not sure we use this notion}Given a multiplayer arena $\gameful$, 
% Formally $\states_1 = \states_i$, $\states_2 = \bigcup_{j\neq i}\states_j$,
% $\obj_1 = \obj_i$, and $\obj_2 = \plays(\game) \setminus \obj_1$.

\paragraph{LTL objective}
We describe specifications using the Linear-time Temporal Logic (LTL). An LTL specification is a formula $\phi$ defined using the following grammar:
\begin{align*}
  \phi & ::= \alpha \mid
  \lnot \phi \mid
  \phi \lor \phi \mid
  \nxt \phi \mid
  \phi  \untl\phi
  \ .
\end{align*}
where $\alpha$ is in $\ap$. \camready{As usual we denote with $\Diamond$ the ``finally''
  operator, defined as $\Diamond \phi = true \untl\phi$.}

LTL formulas are evaluated over plays as follows:
\begin{align*}
  \rho &\models \alpha \text{ iff } \alpha \in \ell(\rho[0]) & 
  \rho & \models \lnot \phi \text{ iff } \rho \not\models \phi \qquad
  \rho \models \phi\lor\psi  \text{ iff } \rho \models \phi \text{
    or } \rho \models \psi 
\ ,\\
  \rho & \models \nxt \phi \text{ iff } \rho[1..] \models \phi
    &
  \rho & \models \phi \untl \psi \text{ iff } \exists i \geq 0,~
         \rho[i..] \models \psi \text{ and } \forall 0 \leq j < i,~
         \rho[j..] \models \phi\ ,&&
\end{align*}
where $\rho \in \states^\omega, \alpha\in\ap, \phi \in \text{LTL} \text{, and
}\psi \in \text{LTL}$.

For an LTL formula $\phi$, we define the set $\out{\phi}$ as the set of plays 
satisfying $\phi$, i.e.,
\begin{align*}
    \out{\phi} = 
    \set{
    \rho \in \states^\omega
    \mid
    \rho \models \phi
    }
\end{align*}
When the objectives are described as LTL formulas, a play $\rho$ satisfies the objective
of player $i$ if $\rho \in \out{\obj_i}$. Similarly we will sometimes write
$\obj_i$ to denote the set $\out{\obj_i}$.

In the sequel, we will use the temporal modality $\Diamond \phi$ as a shortcut for the formula $(\true \untl \phi)$.
%\paragraph{Parity objectives} \todo{Introduce LTL objectives}
%Let $\pi$ in $\plays(\game)$,
%we denote by $\Inf(\pi)$ the set of states 
%occurring infinitely often along $\pi$.
%Let $C$ be a finite subset of $\bbN$, 
%and let $\pr \colon \states \rightarrow C$ be a priority function.
%The parity objective for a game $\game$ equipped with the priority function $\pr$
%is given by the set $\parity$ defined as follows
%\begin{align*}
%  \parity(\game) = 
%  \set{
%    \pi \in \plays(\game) 
%    \mid 
%    \min\{ \pr(s) \mid s
%    \in \Inf(\pi) \} \text{ is even}}  
%\end{align*}

\paragraph{Energy objectives}
Let $\cst\colon \edges \to \bbZ$ be a cost function.
To lighten the notation, we write $\cst(s,t)$ instead of $\cst((s,t))$.
Let $h = \init s_1 \ldots s_n $ be a history in
$\hist(\game)$; we abusively write $\cst(h)$ to mean the extension of
$\cst$ to histories that is
\begin{align*}
  %\cst(h) = \cst(\init,s_1) + \sum_{i = 1}^{n-1} \cst(s_i,s_{i+1})
  \camready{\cst(h) = \sum_{i = 0}^{n-1} \cst(s_i,s_{i+1})}
\end{align*}
The energy objective for a game $\game$ equipped with a cost function $\cst$
is given by the set $\en$ described as follows:
\begin{align*}
  \en(\game) = 
  \set{
    \pi \in \plays(\game) 
    \mid
    \forall i \ge 0,\
    \cst(\pref{\pi}{i}) \ge 0
  }
\end{align*}

%We\todo{unused?} denote by $\lW$ the largest absolute value that appears in $\cst$, i.e.
%\begin{align*}
%  \lW =
%  \max \set{|c| \in \bbZ \mid \exists e \in \edges,\ \cst(e) = c}
%\end{align*}
Throughout the paper, values of $\cst$ are encoded in
binary.
%, thus $\lW$ is exponential in its encoding which is
%$\log(\lW)$

\paragraph{Multi-energy objectives}
Let $\bcstfull$ be a multi dimensional cost function. 
We extend $\bcst$ as expected to histories.

The multi-energy objective for a game $\game$ equipped with a multi dimensional cost function $\bcst$
is given by the set $\multiEn$ described as follows:
\begin{align*}
  \multiEn(\game) = 
  \set{
    \pi \in \plays(\game) 
    \mid
    \forall i \ge 0,\
    \bcst(\pref{\pi}{i}) \ge (0, \ldots, 0)
  }
\end{align*}
We will denote by $\cst_i$ the function obtained by projecting $\bcst$ over the $i$-th dimension.

%\paragraph{Multi energy-LTL objectives}
%Let $\game$ be zero-sum game equipped with both a priority
%unction $\pr$ and a cost function $\bcst$, the energy parity
%bjective $\enPar$ for this game is given by the set
%[
%  \multiEnLTL(\game) =
%  \multiEn(\game) \cap \parity(\game)
%\]
%Given an energy-parity game and a state, 

%\emph{the initial credit} 
%problem asks whether there exists an initial value for the 
%energy vector such that the first player has a strategy to ensure 
%both objectives.

\subsection{Solution concept}
We define in our setting the notion of equilibrium introduced by
Nash. A Nash equilibrium is a profile of strategies in which no player
could do better by unilaterally changing his strategy, provided that
the other players keep their strategies unchanged. 
The set of all the Nash equilibria in a game is denoted $\NE$.

\paragraph{Nash equilibria}
For a multi-player game $\gameful$ with objectives
$\obj_1,\linebreak\ldots, \obj_n$ for each player,
a profile $\fullProf$ is a \emph{Nash equilibrium} ($\NE$)
if for every player~$i$ and every strategy
$\sigma'_i$ for $i$ the following holds true:
\begin{align*}
  \pay_i(\prof) \ge 
  \pay_i(\devProf{\prof_{\coa{i}},\sigma'_i})
\end{align*}

\noindent
\camready{
Equivalently for each player $i$ and for each strategy $\sigma'_i$, if 
$\devProf{\prof_{\coa{i}},\sigma'_i} \in \obj_i$ then
  $\devProf{\prof} \in \obj_i$.}

\subsection{Rational synthesis in the commons}
%\paragraph{Careless cooperative rational synthesis TAKE OUT}
%Let $\gameful$ be a game, $\bcstfull$ be a multi dimensional cost function,
%objectives $\obj_1, \ldots, \obj_n$, $\obj$ a global specification and let
%$\prof$ be a strategy profile. Then $\prof$ is a solution to the
%careless cooperative rational synthesis problem if:
%\begin{align*}
%  &\out{\prof} \in \multiEn(\game) \cap \obj, \text{and }
%   \prof \in \NE
%\end{align*}
%We denote the set of all the solutions by $\NELess(\game)$.

\paragraph{Careful cooperative rational synthesis}
Let $\gameful$ be a game, $\bcstfull$ be a multi dimensional cost function,
objectives $\obj_1, \ldots, \obj_n$, $\obj$ a global specification and let
$\prof$ be a strategy profile. Then $\prof$ is a solution to the
careful cooperative rational synthesis problem if:
\begin{align*}
  & \out{\prof} \in \multiEn(\game) \cap \obj, \text{and }
    \forall \sigma'_i \text{ a strategy for player }i, \\
  & \devProf{\sigma_{-i},\sigma'_i} \in
  \obj_i \cap \multiEn(\game) \implies 
    \out{\prof} \in \obj_i
\end{align*}
%We denote the set of all the solutions by $\NEFul(\game)$.

\pagebreak
\section{Undecidability}
\label{sec:undecidability}

We present multi-counter automata and the problem of reachability which is undecidable. We reduce it to the problem of careful cooperative rational synthesis.

\subsection{Multi-counter automata}
A \emph{$n$-counter automaton} $\Gamma$ is a tuple $(L, \delta, l_0)$ where 
$L$ is a finite set of locations, $\delta$ is a set of transitions, and $l_0 \in L$ is the initial location.
A transition in $\delta$ is a tuple $(l, \vec{w}, \vec{g}, l')$ where
$l$ and $l'$ are locations in $L$, $\vec{w} \in \mathbb{Z}^n$ represents the weights of the transition, and $\vec{g} \in (\mathbb{N} \times (\mathbb{N} \cup \{\omega\}))^n$ represents the guards of the transitions.
Given a transition $\delta$, we note $g_i[lo]$ the lower-bound for counter~$i$ and $g_i[up]$ the upper-bound.

A finite run in a $n$-counter automaton is a triple $(k, \mu_1, \mu_2)$, where $\mu_1: \{0, \ldots, k\} \to L$ and $\mu_2: \{0, \ldots, k\} \to \mathbb{Z}^n$, and such that:
\begin{itemize}
    \item $\mu_1(0) = l_0$ and $\mu_2(0) = (0, \ldots, 0)$
    \item for every $i < k$, if $\mu_1(i) = l$ and $\mu_2(i) = (c_1,\ldots, c_n)$, and $\mu_1(i+1) = l'$ and $\mu_2(i+1) = (c'_1,\ldots, c'_n)$, then there is 
$(l, \vec{w}, \vec{g}, l') \in \delta$, such that for all $0 \leq i \leq n$, we have $g_i[lo] \leq c_i \leq g_i[up]$
and $c'_i = c_i + w_i$.
\end{itemize}

The \emph{reachability problem $(\Gamma, t)$ in $n$-counter automata} asks, given a $n$-counter automaton  $\Gamma = (L, \delta, l_0)$ and a location $t \in L$, whether there is a finite run $(k, \mu_1, \mu_2)$ such that $\mu_1(k) = t$, and $\mu_2(k) = (0, \ldots, 0)$.

%The reachability problem in $2$-counter automata is undecidable~\cite{Minsky}.
\camready{The following lemma can be easily proved using a reduction from $2$-counter machines~\cite{Minsky}.}
%\todo{insert ref possibly \cite{HU79} but not sure about exact def}

\begin{lemma}
The reachability problem in 2-counter automata is undecidable.
\end{lemma}

%\begin{proof}
%The lemma follows by reducing from reachability in two counters Minsky machines. 
%A two counters Minsky machine is a tuple $M = (Q, C_1,C_2, \delta, \mathsf{op})$, where $Q$ is a finite set
%of locations, $C_1$ and $C_2$ are counters with non negative values, and $\delta$ 
%is a transition function defined from $Q\times \mathbb{N} \times \mathbb{N}$ to
%$Q\times \mathbb{N} \times \mathbb{N}$.
%At any transition, a counter automaton can increment a counter, 
%decrement a counter provided that the resulting counter value is at least zero, or
%test if the value of a counter is zero. 
%The reachability problem for Minsky machine asks whether a configuration $(t,0,0)$ can be reached from
%an initial configuration $(s, 0,0)$.
%In the following we write 
%$q \xrightarrow{c \colon = c + z} q'$ to denote a transition from state $q$ to state $q'$ that updates a counter %$c$ by adding the integer $z$.
%%We can simulate the operation on a counter $c$ on transition from $q$ to $q'$ as
%We simulate this transition as follows
%\begin{align*}
%    q \xrightarrow{c \colon = c + z} q' 
%    \text{ becomes }
%    \begin{cases}
%    q \xrightarrow{z, [0, \omega)} q' \text{ if } z \ge 0\\
%    q \xrightarrow{z, [-z, \omega)} q' \text{ if } z < 0    
%    \end{cases}
%\end{align*}
%%\todo{how do we simulate transitions when counters go below 0, since the guards are naturals?}
%For the zero-test transitions from $q$ to $q'$ on counter $a$, we set the guards such that $g_a[lo] = g_a[up]=0$.
%\end{proof}

\subsection{Undecidability of multi-resources careful cooperative rational synthesis}

We reduce the reachability problem in $2$-counter automata into the problem of careful cooperative rational synthesis with two resources and two players.

From a $2$-counter automaton $\Gamma = (L, \delta, l_0)$ and a target location $t$, we are going to build a game\linebreak $\game_{\Gamma, t} = \langle \states_{\Gamma, t}, (\states_1 \biguplus \states_2), \camready{s_{\Gamma, t}}, \{1,2\}, \edges_{\Gamma, t}, \camready{\ap_{\Gamma, t}, \lab_{\Gamma, t}} \rangle$ with costs $\bcst_{\Gamma, t}$ and objectives $\obj_1, \obj_2$, $\obj = \obj_1$, in such a way that a solution to the reachability problem exists iff a solution to the careful cooperative rational synthesis exists.

\paragraph{Construction}
% AGENTS
We are going to use two players in this construction. Player~$1$'s role will be to build a solution, choosing the transitions to follow. Player~$2$'s role will be to ``check'' that the transitions are legitimate, making Player~$1$ fail in his tasks if a transition that does not respect the guards is taken.

% WIN STATES OF THE GAME
In $\game_{\Gamma, t}$, we first add the three states: $\winone$ representing the winning state of Player~$1$, and $\wintwo$, and $\wintwo'$ representing the winning states of Player~$2$. They will be sink states, and it does not matter who controls them.

% LOCATIONS
Each location in $\Gamma$ is also a state in $\game_{\Gamma, t}$, controlled by Player~$1$.

% TRANSITIONS
For each transition $\tau = (l, \vec{w}, \vec{g}, l')$ in $\delta$ we introduce two states $\tau_<$ and $\tau_>$, both controlled by Player~$2$, and a few transitions. 
Intuitively, the state $\tau_>$ will serve as a state in which Player~$2$ will ``check'' that the upper-guard is satisfied. (Player~$2$ will have the opportunity to win if it does not.) The state $\tau_<$ will serve for the system to ``check'' that the lower guard is satisfied. (The value of one of the resources will go below zero if it is not the case.)
There are four cases to consider; They are illustrated on Figure~\ref{fig:gadgets-transitions}: (\ref{fig:ff})~$g_1[up] \not= \omega$ and $g_2[up] \not= \omega$; (\ref{fig:fi})~$g_1[up] \not= \omega$ and $g_2[up] = \omega$; (\ref{fig:if})~$g_1[up] = \omega$ and $g_2[up] \not= \omega$; (\ref{fig:ii})~$g_1[up] = \omega$ and $g_2[up] = \omega$.

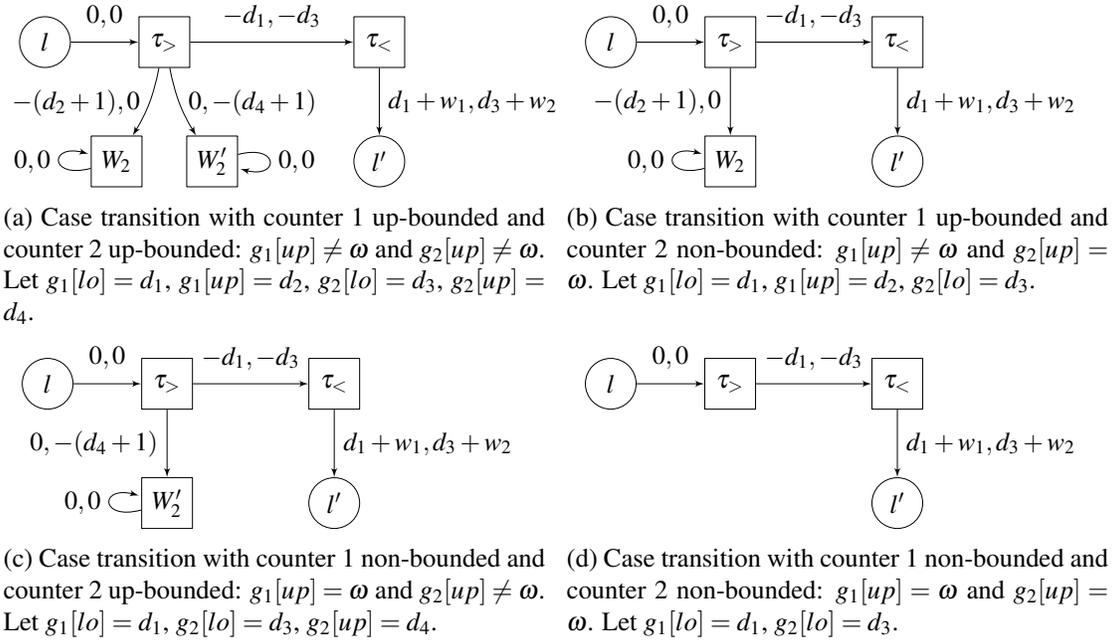
\begin{figure}
    \centering
    \begin{subfigure}[t]{0.45\textwidth}
        %\centering
        \scalebox{0.90}{% reduction 2-counter reachability to 2-player 2 resources careful cooperative rational synthesis
% case transition with counter 1 top-bounded and counter 2 top-bounded

\begin{tikzpicture}[>=latex', join=bevel, initial text = , every node/.style=,
  minimum size=.75cm]
  % States
  \node (l) at (0bp, 0bp) [draw, circle]{$l$};
  \node (lmore) at (50bp, 0bp) [draw,  rectangle]{$\tau_>$};

  \node (win2) at (30bp, -50bp) [draw, rectangle]{$\wintwo$};
  \node (win22) at (70bp, -50bp) [draw, rectangle]{$\wintwo'$};
  
  \node (lless) at (140bp, 0bp) [draw, rectangle]{$\tau_<$};
  \node (lprime) at (140bp, -50bp) [draw, circle]{$l'$};
  % Edges
  \draw[->] (l) to node [above] {$0,0$} (lmore);
  \draw[->] (lmore) to node [above] {$-d_1,-d_3$} (lless);
  \draw[->] (lmore) [bend left=10] to node [left] {$-(d_2+1),0$} (win2);
  \draw[->] (lmore) [bend right=10] to node [right] {$0, -(d_4+1)$} (win22);
  \draw[->] (lless) to node [right] {$d_1+w_1, d_3+w_2$} (lprime);
  % Selfloops
  \draw (win2) edge [loop left] node [left] {$0,0$} ();
  \draw (win22) edge [loop right] node [right] {$0,0$} ();

\end{tikzpicture}}
        \caption{Case transition with counter~1 up-bounded and counter~2 up-bounded: $g_1[up] \not= \omega$ and $g_2[up] \not= \omega$. Let $g_1[lo] = d_1$, $g_1[up] = d_2$, $g_2[lo] = d_3$, $g_2[up] = d_4$.\label{fig:ff}}
    \end{subfigure}
    ~
    \begin{subfigure}[t]{0.45\textwidth}
        \centering
        \scalebox{0.90}{% reduction 2-counter reachability to 2-player 2 resources careful cooperative rational synthesis
% case transition with counter 1 top-bounded and counter 2 non bounded

\begin{tikzpicture}[>=latex', join=bevel, initial text = , every node/.style=,
  minimum size=.75cm]
  % States
  \node (l) at (0bp, 0bp) [draw, circle]{$l$};
  \node (lmore) at (50bp, 0bp) [draw,  rectangle]{$\tau_>$};
  \node (win2) at (50bp, -50bp) [draw, rectangle]{$\wintwo$};
  \node (lless) at (120bp, 0bp) [draw, rectangle]{$\tau_<$};
  \node (lprime) at (120bp, -50bp) [draw, circle]{$l'$};
  % Edges
  \draw[->] (l) to node [above] {$0,0$} (lmore);
  \draw[->] (lmore) to node [above] {$-d_1,-d_3$} (lless);
  \draw[->] (lmore) to node [left] {$-(d_2+1),0$} (win2);
  \draw[->] (lless) to node [right] {$d_1+w_1, d_3+w_2$} (lprime);
  % Selfloops
  \draw (win2) edge [loop left] node [left] {$0,0$} ();

\end{tikzpicture}}
        \caption{Case transition with counter~1 up-bounded and counter~2 non-bounded: $g_1[up] \not= \omega$ and $g_2[up] = \omega$. Let $g_1[lo] = d_1$, $g_1[up] = d_2$, $g_2[lo] = d_3$.\label{fig:fi}}
    \end{subfigure}

    \begin{subfigure}[t]{0.45\textwidth}
        \centering
        \scalebox{0.90}{% reduction 2-counter reachability to 2-player 2 resources careful cooperative rational synthesis
% case transition with counter 1 non bounded and counter 2 top-bounded

\begin{tikzpicture}[>=latex', join=bevel, initial text = , every node/.style=,
  minimum size=.75cm]
  % States
  \node (l) at (0bp, 0bp) [draw, circle]{$l$};
  \node (lmore) at (50bp, 0bp) [draw,  rectangle]{$\tau_>$};
  \node (win2) at (50bp, -50bp) [draw, rectangle]{$\wintwo'$};
  \node (lless) at (120bp, 0bp) [draw, rectangle]{$\tau_<$};
  \node (lprime) at (120bp, -50bp) [draw, circle]{$l'$};
  % Edges
  \draw[->] (l) to node [above] {$0,0$} (lmore);
  \draw[->] (lmore) to node [above] {$-d_1,-d_3$} (lless);
  \draw[->] (lmore) to node [left] {$0, -(d_4+1)$} (win2);
  \draw[->] (lless) to node [right] {$d_1+w_1, d_3+w_2$} (lprime);
  % Selfloops
  \draw (win2) edge [loop left] node [left] {$0,0$} ();

\end{tikzpicture}}
        \caption{Case transition with counter~1 non-bounded and counter~2 up-bounded: $g_1[up] = \omega$ and $g_2[up] \not= \omega$. Let $g_1[lo] = d_1$, $g_2[lo] = d_3$, $g_2[up] = d_4$.\label{fig:if}}
    \end{subfigure}
    ~
    \begin{subfigure}[t]{0.45\textwidth}
        \centering
        \scalebox{0.90}{% reduction 2-counter reachability to 2-player 2 resources careful cooperative rational synthesis
% case transition with counter 1 non bounded and counter 2 non-bounded

\begin{tikzpicture}[>=latex', join=bevel, initial text = , every node/.style=,
  minimum size=.75cm]
  % States
  \node (l) at (0bp, 0bp) [draw, circle]{$l$};
  \node (lmore) at (50bp, 0bp) [draw,  rectangle]{$\tau_>$};
  \node (lless) at (120bp, 0bp) [draw, rectangle]{$\tau_<$};
  \node (lprime) at (120bp, -50bp) [draw, circle]{$l'$};
  % Edges
  \draw[->] (l) to node [above] {$0,0$} (lmore);
  \draw[->] (lmore) to node [above] {$-d_1,-d_3$} (lless);
  \draw[->] (lless) to node [right] {$d_1+w_1, d_3+w_2$} (lprime);
  % Selfloops
\end{tikzpicture}}
        \caption{Case transition with counter~1 non-bounded and counter~2 non-bounded: $g_1[up] = \omega$ and $g_2[up] = \omega$. Let $g_1[lo] = d_1$, $g_2[lo] = d_3$.\label{fig:ii}}
    \end{subfigure}
    \caption{Gadgets to encode the transitions.\label{fig:gadgets-transitions}}
\end{figure}

% TARGET LOCATION
We also need a gadget to ``check'' that solutions are runs that reach the target location with the values of the counters being zero. We introduce a state $t_?$ in $\game_{\Gamma, t}$, controlled by Player~$2$. See Figure~\ref{fig:gadget-target}.

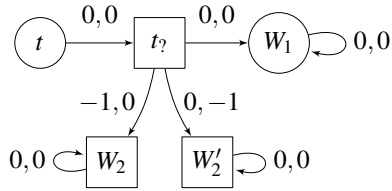
\begin{figure}[t]
\centering
\scalebox{0.90}{\begin{tikzpicture}[>=latex', join=bevel, initial text = , every node/.style=,
  %scale=0.75                                                                                                                                                                            
  minimum size=.75cm]
  % States                                                                                                                                                                               
  \node (t) at (0bp, 0bp) [draw, circle]{$t$};
  \node (ttest) at (50bp, 0bp) [draw,  rectangle]{$t_?$};
  \node (win2) at (30bp, -50bp) [draw, rectangle]{$\wintwo$};
  
  \node (win22) at (70bp, -50bp) [draw, rectangle]{$\wintwo'$};

  \node (win1) at (100bp, 0bp) [draw, circle]{$\winone$};
  % Edges                                                                                                                                                                                
  \draw[->] (t) to node [above] {$0,0$} (ttest);
  \draw[->] (ttest) to node [above] {$0,0$} (win1);
  \draw[->] (ttest) [bend left=10] to  node [left] {$-1,0$} (win2);
  \draw[->] (ttest) [bend right=10] to  node [right] {$0,-1$} (win22);

  % Selfloops                                                                                                                                                                            
  \draw (win1) edge [loop right] node [right] {$0,0$} ();
  \draw (win2) edge [loop left] node [left] {$0,0$} ();
  \draw (win22) edge [loop right] node [right] {$0,0$} ();

\end{tikzpicture}}
\caption{Gadget to encode the target location $t$ with counter values $(0,0)$.\label{fig:gadget-target}}
\end{figure}

% COST FUNCTION
The five gadgets also completely specify the cost function $\bcst_{\Gamma, t}$.

% THAT'S ALL FOLKS.
The game $\game_{\Gamma, t}$ does not contain any other state or transition. 

% init, AP, label
\camready{The initial state is $s_{\Gamma, t}$. The set of propositions $\ap_{\Gamma, t}$ is $S_{\Gamma, t}$, and the labeling function $\lab_{\Gamma, t}$ is the identity.}

\medskip
% OBJECTIVES
In the game so obtained, the objective of Player~$1$ is to reach the state $\winone$ and the objective of Player~$2$ is to reach state $\wintwo$ or $\wintwo'$. % Since $\winone$, $\wintwo$ and $\wintwo'$ are sink states, with B\"uchi objectives or parity objectives it suffices to make it the agent's objective to visit their respective winning states infinitely many times.

\medskip
We now prove that the construction above can serve as a reduction from the reachability problem in $2$-counter automata into the problem of careful rational synthesis (with two common resources).
\begin{proposition}
Let $\Gamma = (L, \delta, l_0)$ be a $2$-counter automaton and let $t$ be a location in $L$. The reachability problem $(\Gamma, t)$ has a positive answer iff there is a solution to the careful cooperative rational synthesis in the game 
$\game_{\Gamma, t} = \langle \states_{\Gamma, t}, (\states_1 \biguplus \states_2), \camready{s_{\Gamma, t}}, \{1,2\}, \edges_{\Gamma, t}, \camready{\ap_{\Gamma, t}, \lab_{\Gamma, t}} \rangle$.
%, reachabilty objectives $\obj_1 = \states^*_{\Gamma, t} \winone \states^\omega_{\Gamma, t}$, $\obj_2 = \states^*_{\Gamma, t} (\wintwo \cup \wintwo') \states^\omega_{\Gamma, t}$, and $\obj = \obj_1$.
\end{proposition}

\begin{proof}
\noindent\textbf{Left to right.} Suppose there is a solution to the reachability problem.
There is a run $(k, \mu_1, \mu_2)$ such that $\mu_1(k) = t$, and $\mu_2(k) = (0, 0)$. For every $i$, let 
$\tau^i = (l^i, (w^i_1, w^i_2), (g^i_1, g^i_2), l'^i)$ be the transition between $\mu_1(i)$ and $\mu_1(i+1)$. By definition of the reachability problem, we know that the guards of all of them are satisfied.
\camready{Since it reaches $t$ at step $k$, 
the following sequence is a run in $\game_{\Gamma, t}$: 
$\pi_\rho := \mu_1(0) \cdot \tau^0_> \cdot \tau^0_< \cdot \mu_1(1) \cdot  \tau^1_> \cdot \tau^1_< \cdots \mu_1(k) \cdot t_? \cdot (\winone)^\omega$.}
We argue that $\pi_\rho$ satisfies the objective, it never depletes the common resources, it is the play of a Nash equilibrium.
\begin{itemize}
% SATISFIES OBJ
\item
Since the play $\pi_\rho$ enters $\winone$, it is in $\obj$.

% NEVER DEPLETES ENERGY LEVELS
\item
Since $(k, \mu_1, \mu_2)$ is a solution to the reachability problem, by construction of $\game_{\Gamma, t}$, the values of the counters along the play $\pi_\rho$ never go below $0$.

% THE OUTCOME OF A NASH EQUILIBIRUM
\item
% PLAYER 2 HAS NO INCENTIVE TO DEVIATE
Along the play $\pi_\rho$, in every $\tau^i_<$, Player~$2$ chooses to go $\tau^i_>$. Since the guards are respected along the run $(k, \mu_1, \mu_2)$ in $\Gamma$, by construction of $\game_{\Gamma, t}$, Player~$2$ never has an opportunity in any state $\tau^i_>$ to deviate carefully (and profitably) to $\wintwo$ or $\wintwo'$.

In state $t_?$, Player~$2$ chooses to go to $\winone$. Since $\mu_2(k) = (0,0)$, the play $\pi_\rho$ enters the state $t_?$ with both counters being $0$. Thus, Player~$2$ cannot deviate carefully (and profitably) to $\wintwo$ or $\wintwo'$.

% PLAYER 1 HAS NO INCENTIVE TO DEVIATE
Since the play $\pi_\rho$ enters $\winone$, it is winning for Player~$1$, who has no incentive to deviate.

Hence, $\pi_\rho$ is the outcome of a Nash equilibrium.
\end{itemize}

% GATHER EVERYTHING
So there is a solution to the careful cooperative rational synthesis.

\noindent\textbf{Right to left.} Suppose there is a solution to the problem of careful cooperative rational synthesis.

By definition of the problem of careful cooperative rational synthesis, there is a strategy profile $\prof$ such that:
\begin{enumerate}
    \item \label{en:ne} The profile $\prof$ is a Nash equilibrium.
    \item \label{en:good-energy} The value of each resource never goes below $0$.
    \item \label{en:obj} The play $\play{\prof}$ reaches the state $\winone$.
\end{enumerate}

We argue that the play $\play{\prof}$ of $\game_{\Gamma, t}$ reaches $t$ with the value of the resources being $(0,0)$.
\begin{enumerate}
\setcounter{enumi}{3}
\item \label{en:play-reaches-t-00}
Since $\play{\prof}$ reaches the state $\winone$, by construction of the game, $\play{\prof}$ is losing for Player~$2$. By construction also, $\play{\prof}$ goes through the state $t_?$. Since $\prof$ is a Nash equilibrium, it must be that the value of the resources when it goes through the state $t_?$ are $(0,0)$, otherwise, Player~$2$ would profitably choose to go to $\wintwo$ or $\wintwo'$ instead of $\winone$.
Suppose $\play{\prof}$ reaches $t$ at index $k_t$. I.e., $\play{\prof}[k_t] = t$, and $\bcst(\pref{\play{\prof}}{k_t}) = (0,0)$.
\end{enumerate}

Let $\rho_{\prof} = (k_t, \mu_1, \mu_2)$ be the finite run in $\Gamma$ that is the projection of $\pref{\play{\prof}}{k_t}$ onto $L$. We argue that $\rho_{\prof}$ is a solution to the reachability problem $(\Gamma, t)$. 
\begin{itemize}
    \item From item~\ref{en:play-reaches-t-00}, $\play{\prof}$ reaches $t$ at index $k_t$ with resource values $(0,0)$. So $\rho_{\prof}$ reaches the state $(t, (0,0))$.
    \item We argue that the upper guards are always respected along $\rho_{\prof}$. By item~\ref{en:obj} and by construction, we know that $\play{\prof}$ is not winning for Player~$2$. But since $\prof$ is a Nash equilibrium (item~\ref{en:ne}), Player~$2$ never has an opportunity to carefully (and profitably) deviate to $\wintwo$ or $\wintwo'$. So when going through a transition $\tau = (l, \vec{w}, \vec{g}, l')$, say at step $j$ along $\rho_{\prof}$, if $\mu_2(j) = (c_1,c_2)$ then $c_1 \leq g_1[up]$ and $c_2 \leq g_2[up]$.
    \item By construction, it follows from item~\ref{en:good-energy} that the lower guards are always respected along $\rho_{\prof}$.
\end{itemize}
\end{proof}

The next result follows at once.
\begin{theorem}
The problem of careful cooperative rational synthesis is undecidable, even with two players and two resources, and reachability objectives.
\end{theorem}

\section{Bounded resources and decidability}
\label{sec:bounded-resources}

In this section we consider the so-called bounded setting. Here each counter will be bounded from above by some bound $B$ and cannot store more that $B$. Intuitively, one can continue to charge a battery, but the energy exceeding the capacity will be lost in heat; just like one can continue to fill up a tank of water but it will spill over when the capacity is reached.

The main result of this section will be the decidability of the synthesis problem in this case.

Given a game $\gameful$, and multidimensional cost function $\bcstfull$, we fix a vector in $\bbN^d$ representing the maximal capacity of each counter. \camready{We use $\bcst(s,s')[i]$ to denote the $i$-th component of the tuple
$\bcst(s,s')$.}
%between $s_{i}$ and $s'_{i}$. 
Along a run when a counter
is at capacity its value cannot increase.
Formally, Assume that the capacity is given by the following vector
$\vec{B}= (B_1, \ldots, B_d)$. We define the operator $\oplus_{\vec{B}}$ over vectors in $\bbZ^d$ as follows:
 \begin{align*}
     (c_1, \ldots, c_d ) \oplus_{\vec{B}} (c'_1, \ldots, c'_d ) = (x_1, \ldots, x_d) 
     \text{ where } \forall 1\le i \le d,~ x_i = \min(c_i + \bcst(s,s')[i], B_i)
 \end{align*}
 %$(c_1, \ldots, c_d ) \oplus_{(B_1, \ldots, B_d} \bcst(s,s') = (c'_1, \ldots, c'_d)$ where $\oplus_{(B_1, \ldots, B_d}$ is $min(c_i + \bcst(s,s')[i], B_i)$ on every dimension $i$...?
 
We can now define the cost vector along a history $h = \init\ s_1\ldots\ s_n$ inductively as follows
\begin{align*}
  \bcst(h) = \bcst(\init\ s_1\ldots\ s_{n-1}) \oplus_{\vec{B}} \bcst(s_j,s_{j+1})
\end{align*}

The decidability result is obtained through an unfolding of the arena. This unfolding 
constructs a multiplayer game without costs $\uGame$ where the set of states is 
$\uStates = \states \times \set{0,\ldots, B_1} \times \ldots \times \set{0,\ldots, B_d} \cup \set{\bot}$.\\
The set of edges is $\uEdges$ in $(\uStates \times \uStates) \cup (\uStates\times \set{\bot}) \cup \set{(\bot, \bot)}$ and is 
defined as follows:
\begin{align*}
    ((s, c_1, \ldots, c_d ),(s', c'_1, \ldots, c'_d)) \in \uEdges
\end{align*}
if
\begin{itemize}
\item $(s,s')$ in $\edges$
\item $(c_1, \ldots, c_d ) \oplus_{\vec{B}} \bcst(s,s') = (c'_1, \ldots, c'_d)$
\item $(c'_1, \ldots, c'_d) \ge (0, \ldots, 0)$
\end{itemize}
and 
\begin{align*}
    ((s, c_1, \ldots, c_d ),\bot)\in \uEdges
\end{align*}
if 
\begin{itemize}
\item there exists $s$ such that $(s,s')$ in $\edges$
\item for some counter value $c_i$ we have $c_i + \cst_i(s,s') < 0$.
\end{itemize}
\camready{Also $(\bot,\bot) \in \uEdges$.}

In this new game a state $(s, c_1, \ldots, c_d)$ belongs to player $i$ if $s$ belongs to player $i$.
Player 1 controls also the fresh state $\bot$.
The objective of each player is the same and the global specification in $\uGame$ is $\obj \land \lnot\Diamond \bot$.
Plays in this unfolding are infinite sequences of $\uStates$. In order to relate plays in $\game$ with plays in $\uGame$
we use the following projection $\pi$ 
\camready{defined over the set of histories as follows: first,
\begin{align*}
    \pi(\init) = 
      (\init,0,\ldots, 0)
\end{align*}
and for a history $h=\init\ s_1\ \ldots\ s_l$ in $\game$:
\begin{align*}
    \pi(\init\ s_1\ \ldots\ s_l) = \pi(\init\ s_1\ \ldots\ s_{l-1}) 
    (s,\bcst(\init\ s_1\ \ldots\ s_l))
\end{align*}}
We extend $\pi$ over the plays as expected and denote by $\pi^{-1}$ the inverse mapping.

For a play $\hat{\rho}$, we say that $\hat{\rho}$ satisfies the objective of player $i$ if $\pi^{-1}(\hat{\rho})$ satisfies $\obj_i$.
We will say that a play $\hat{\rho}$ satisfies $\obj \land \lnot\Diamond \bot$
if $\pi(\hat{\rho})$ satisfies $\obj$ and $\hat{\rho}$ satisfies $\lnot\Diamond \bot$.

\begin{proposition}\label{prop:correctness}
 There exists a solution to the careful synthesis if and only if 
 there exists a Nash equilibrium in the unfolding whose outcome satisfies $\obj \land \lnot\Diamond \bot$.
\end{proposition}

\begin{proof}
Let $\prof$ be a solution in $\game$, we construct $\uProf$ as follows:
\begin{align*}
    \uProf = (\widehat{\sigma}_i, \ldots, \widehat{\sigma}_n)
\end{align*}
such that each $\widehat{\sigma_i}$ is defined as follows:
\begin{align*}
    \widehat{\sigma}_i (\hat{h}) = \sigma_i(\pi(\hat{h}))
\end{align*}
where $\hat{h}$ is a history of $\uGame$, and $\sigma_i$ is the strategy of player $i$ in the profile $\prof$.
We argue that $\uProf$ is also a solution thanks to the following fact:
\begin{itemize}
    \item $\uProf$ is a Nash equilibrium since each $\widehat{\sigma}_i$ ensures the same payoff as $\sigma_i$.
    \item $\prof$ is solution, hence it ensures that the energy along its outcome never drops bellow 0 for all the counters, 
    hence by construction $\bot$ is never visited.
\end{itemize}

Let $\uProf$ be a solution  of $\uGame$, then we construct $\prof$ as follows:
\begin{align*}
    \prof = (\sigma_1, \ldots, \sigma_n)
\end{align*}
where for each history $h$,
\begin{align*}
    \sigma_i = \widehat{\sigma}_i(\pi(h))
\end{align*}
We argue that $\prof$ is also a solution thanks to the following fact:
\begin{itemize}
    \item $\prof$ is a Nash equilibrium since each $\sigma_i$ ensures the same payoff as $\widehat{\sigma}_i$.
    \item $\uProf$ is solution, hence it ensures that $\bot$ is never visited, therefore by construction
        the energy along the outcome of $\prof$ never drops below 0 for all the counters,
\end{itemize}
\end{proof}

\camready{By Proposition~\ref{prop:correctness}, we know that we can solve careful synthesis in the original arena by reducing it to plain rational synthesis in the unfolding. 
It is readily seen that the size of the unfolding of the arena defined above is exponential in the size of the original arena.
On the other hand, solving the cooperative rational synthesis with LTL objectives is in $\twoexptime$ in the size of the objectives formulas, and polynomial in the size of the arena~\cite{FKL10,KuSh22}. It follows that our problem is $\twoexptime$ when the counters are bounded.}
\begin{theorem}
The careful cooperative rational synthesis is $\twoexptime$-complete when the counters are bounded.
\end{theorem}

%Now notice that solving the cooperative rational synthesis with LTL objective is a $\twoexptime$ but since it is simply exponential in the size of the game,
%it follows that our problem is $\twoexptime$ when the counters are bounded.

\section{Conclusion}
As agents and robots are always more likely to roam the physical world, the formal tools to engineer them need to take into account the resource-sensitiveness of their activities.

We presented a model for autonomous and rational agents interacting in environments with multiple resources.
We focused on a problem of rational planning, rational synthesis, that consists in finding a non-cooperative equilibrium (a Nash equilibrium) that satisfies a system objective, and never depletes the resources.
We showed that this problem is undecidable.

We then proposed a variant where the storage capacity is bounded for all resources. We claim that this is promising for the applicability to real-world settings. The  storage of resources is indeed generally limited (energy power capacity of a battery, volume of a water tank, etc). Moreover, we proved that the problem of rational synthesis with LTL objectives becomes decidable in double-exponential time, which is no harder than plain controller synthesis for LTL specifications.

\medskip

In the future, we are interested in the study of problems to elaborate tools that better equip the engineers of agents and robots in resource-sensitive environments. In particular, we will investigate 
problems of \emph{parameterized} synthesis allowing an engineer to partially model a system, leaving some quantities unspecified (for example, we can leave unspecified some 
weights of transitions or the bounds of the resources), and with the aim of automatically completing the system in a way that it admits a solution to the synthesis problem.

%\nocite{*}
\bibliographystyle{eptcs}
\bibliography{ref}
\end{document}